\documentclass[11pt]{article}

\usepackage{times}
\usepackage{abstract}
\usepackage[margin=.75in]{geometry}
\usepackage{amsmath}
\usepackage{amssymb}
\usepackage[numbers,sort&compress]{natbib}

\usepackage[pagebackref]{hyperref}
\usepackage{amsthm}
\newtheorem{theorem}{Theorem}[section]
\newenvironment{theorem_cite}[1]{\begin{theorem} {\rm (#1)}}{\end{theorem}}

\usepackage{aliascnt}
\usepackage[nameinlink]{cleveref}
\newaliascnt{lemma}{theorem}
\newtheorem{lemma}[lemma]{Lemma}
\aliascntresetthe{lemma}
\crefname{lemma}{lemma}{lemmas}
\newaliascnt{corollary}{theorem}
\newtheorem{corollary}[corollary]{Corollary}
\aliascntresetthe{corollary}
\crefname{corollary}{corollary}{corollaries}
\theoremstyle{definition}
\newaliascnt{definition}{theorem}
\newtheorem{definition}[definition]{Definition}
\aliascntresetthe{definition}
\crefname{definition}{Definition}{Definitions}
\Crefname{definition}{Definition}{Definitions}

\crefalias{definition_cite}{definition}
\newaliascnt{remark}{theorem}
\newtheorem{remark}[remark]{Remark}
\aliascntresetthe{remark}
\crefname{remark}{remark}{remarks}
\usepackage{thm-restate}

\usepackage{amsfonts}
\usepackage{amsmath}
\usepackage{amssymb}
\usepackage{enumerate}
\usepackage{xspace}

\newcommand{\CCfont}[1]{\ensuremath{\mathsf{#1}}}

\newcommand{\pair}[1]{\langle #1 \rangle}

\newcommand{\prprefix}{\underset{\neq}{\sqsubset}}

\newcommand{\p}{{\CCfont{p}}}
\newcommand{\pspace}{{\CCfont{pspace}}}
\newcommand{\all}{{\CCfont{all}}}

\newcommand{\C}{\mathbf{C}}

\newcommand{\NP}{\CCfont{NP}}
\newcommand{\SIZE}{\CCfont{SIZE}}
\newcommand{\ZPP}{\CCfont{ZPP}}
\newcommand{\E}{\CCfont{E}}
\newcommand{\ESPACE}{\CCfont{ESPACE}}

\renewcommand{\P}{\CCfont{P}}
\newcommand{\DeltaE}{\Delta^\E}
\newcommand{\DeltaEthree}{\DeltaE_3}
\newcommand{\SPtwo}{\CCfont{S}^\P_2}

\newcommand{\poly}{\CCfont{poly}}
\newcommand{\Ppoly}{{\P/\poly}}

\newcommand{\FP}{\CCfont{FP}}

\newcommand{\calA}{{\cal A}}

\renewcommand{\FP}{\CCfont{FP}}

\newcommand{\twonn}{\frac{2^n}{n}}
\newcommand{\lutzbound}{\twonn\left(1+\frac{\alpha \log n}{n}\right)}

\newcommand{\FZPP}{\CCfont{FZPP}}

\newcommand{\FSPtwo}{\CCfont{FS}^\P_2}
\newcommand{\SEtwo}{\CCfont{S}^\E_2}

\renewcommand{\varepsilon}{\epsilon}

 \renewcommand{\C}{\CCfont{C}}

\title{Exponential-Size Circuit Complexity is\\Comeager in Symmetric Exponential Time}

\author{John M. Hitchcock\thanks{Department of Electrical Engineering and Computer Science, University of Wyoming. jhitchco@uwyo.edu. This research was supported in part by NSF grant 2431657.}}

\date{}

\begin{document}
\maketitle
\begin{abstract}
	Lutz (1987) introduced resource-bounded category and showed the circuit size class $\SIZE(\twonn)$ is meager within $\ESPACE$. Li (2024) established that the symmetric alternation class $\SEtwo$ contains problems requiring circuits of size $\twonn$.

	In this note, we extend resource-bounded category to $\SEtwo$ by defining meagerness relative to single-valued $\FSPtwo$ strategies in the Banach-Mazur game. We show that Li’s $\FSPtwo$ algorithm for the Range Avoidance problem yields a winning strategy, proving that $\SIZE(\twonn)$ is meager in $\SEtwo$.

	Consequently, languages requiring exponential-size circuits are comeager in $\SEtwo$: they are typical with respect to resource-bounded category.

\end{abstract}
\section{Introduction}

Resource-bounded category, introduced by Lutz~\cite{Lutz87,Lutz:thesis,Lutz:CMCC}, provides a way to analyze the typical properties of complexity classes using topological notions of size. Lutz showed that the class $\SIZE(\twonn)$ of languages decidable by circuits of size $\twonn$ is \emph{meager}
within the exponential-space class $\ESPACE$. \begin{theorem_cite}{Lutz~\cite{Lutz87,Lutz:thesis,Lutz:CMCC}}
	$\SIZE(\twonn)$ is meager in $\ESPACE$.
\end{theorem_cite}

The symmetric exponential-time class $\SEtwo$ is an exponential version of the class $\SPtwo$~\cite{Canetti96,RussellSundaram98} and it lies between $\E$ and $\ESPACE$. Building on work by Korten~\cite{Korten22} and Chen, Hirahara, and Ren~\cite{ChenHiraharaRen24}, Li~\cite{Li24} recently proved that $\SEtwo$ contains languages that require near-maximum circuit size.
\begin{theorem_cite}{Li~\cite{Li24}}\label{th:Li_intro}
	$\SEtwo\not\subseteq\SIZE(\twonn)$.
\end{theorem_cite}

\hyphenation{ex-po-nen-tial-size}
This solved a longstanding open problem to show the second level of the exponential-time hierarchy requires exponential-size circuits~\cite{Kann82,BuFoTh98,Mayo94,MiViWa99,Hitchcock:DERC,Aaronson:NECLB,Aydinlioglu:DAMG}.
Li's proof technique involves constructing a hard language using a \emph{single-valued} $\FSPtwo$ algorithm (a function problem analogue of $\SPtwo$). Given these developments, a natural question is whether the category framework can be refined to operate within $\SEtwo$. Can we show that languages requiring large circuits are not just present, but are in fact \emph{typical} within $\SEtwo$?

In this note, we answer this question affirmatively. We adapt the resource-bounded category machinery to the class $\SEtwo$ by defining meagerness relative to single-valued $\FSPtwo$-computable constructors in the Banach-Mazur game.
Our main result shows that, analogous to Lutz's result for $\ESPACE$, small circuits are atypical within $\SEtwo$.
\begin{restatable}{theorem}{MainTheorem}\label{th:main}
	$\SIZE(\twonn)$ is meager in $\SEtwo$.
\end{restatable}
We prove this by demonstrating that Li's single-valued $\FSPtwo$ algorithm for the range avoidance problem \cite{Korten22,ChenHiraharaRen24} provides a winning strategy against $\SIZE(\twonn)$ in the appropriate Banach-Mazur game.  Thus, languages requiring exponential-size circuits are not just present in $\SEtwo$—they are typical in the sense of resource-bounded category: the class of such languages is comeager in $\SEtwo$.

This paper is organized as follows. We present preliminaries on resource-bounded category and symmetric alternation in \Cref{sec:prelim}. Our results are in \Cref{sec:category_in_S2E}. We conclude in \Cref{sec:conclusion} with an open question about extending our results to resource-bounded measure.

\section{Preliminaries}\label{sec:prelim}

\subsection{Resource-Bounded Category}\label{sec:category}

Resource-bounded category was introduced by Lutz \cite{Lutz87,Lutz:thesis,Lutz:CMCC}.
The {\em Cantor space} $\C$ is the set of all infinite binary
sequences.  A {\em language} (or {\em decision problem}) is a subset
of $\{0,1\}^*$.  We identify each language with the element of Cantor
space that is its characteristic sequence according to the standard
enumeration $s_0 = \lambda, s_1 = 0, s_2 = 1, s_3 = 00, \ldots$ of $\{0,1\}^*$.  In this way, complexity classes (sets of languages) are viewed as subsets of Cantor space.

Baire category classifies sets into two types: {\em first category}
and {\em second category}.  First category sets are also commonly
called {\em meager}.  A set is meager if it is a countable union of
nowhere dense sets.  An equivalent definition comes from Banach-Mazur
games using functions called constructors.

\begin{definition}
	A {\em constructor} is a total, single-valued function $\delta : \{0,1\}^* \to \{0,1\}^*$ which satisfies $x \prprefix \delta(x)$ for all $x \in \{0,1\}^*$.
	The {\em result} of a constructor is the unique sequence $R(\delta)
		\in \C$ that extends $\delta^{(n)}(\lambda)$ for all $n$. \end{definition}

Let $X \subseteq \C$ and let $\Gamma_{\rm I}$ and $\Gamma_{\rm II}$ be
two classes of functions.  In the {\em Banach-Mazur game}
$G[X;\Gamma_{\rm I},\Gamma_{\rm II}]$ there are two players I and
II.  A {\em strategy} in the game is a constructor.  In a play of
the game, player I chooses a strategy $\gamma \in \Gamma_{\rm I}$ and
player II chooses a strategy $\delta \in \Gamma_{\rm II}$.  The {\em
		result} of this play is the sequence $R(\gamma,\delta) = R(\delta \circ \gamma)$.
Intuitively, the result is the sequence obtained when the two
players start with the empty string and take turns extending it with
their strategies.  A {\em winning strategy} for player II is a
strategy $\delta \in \Gamma_{\rm II}$ such that for every $\gamma \in
	\Gamma_{\rm I}$, $R(\gamma,\delta) \not\in X$.

\begin{theorem_cite}{Banach and Mazur}
	A class $X \subseteq \C$ is meager if and only if player II has a
	winning strategy in the game $G[X;\all,\all]$.
\end{theorem_cite}

In resource-bounded category, $\Delta$ denotes a {\em resource bound}
\cite{Lutz:AEHNC,Lutz:CMCC}.  Examples of $\Delta$ include:
\begin{eqnarray*}
	\all &=& \{ f \mid f : \{0,1\}^* \to \{0,1\}^* \}\\
	\p &=& \{ f \mid f\textrm{ is polynomial-time computable} \}\\
	\pspace &=& \{ f \mid f\textrm{ is polynomial-space computable} \}\end{eqnarray*}
\noindent
For a resource bound $\Delta$, we define the {\em result class}
$$R(\Delta) = \{ R(\delta) \mid \delta \in \Delta\textrm{ is a constructor} \}.$$
Then $R(\all) = \C$, $R(\p) = \E$,
and $R(\pspace) = \ESPACE$ \cite{Lutz:CMCC}. Resource-bounded category \cite{Lutz:CMCC} is defined by requiring
player II's winning strategy to be computable within a resource bound $\Delta$.

\begin{definition}\label{def:resource_bounded_category}
	Let $X \subseteq \C$ and let $\Delta$ be a resource bound.
	\begin{enumerate}
		\item  $X$ is {\em $\Delta$-meager} if player II has a winning
		      strategy in the game $G[X;\all,\Delta]$.
		\item  $X$ is {\em $\Delta$-comeager} if $X^c$ is $\Delta$-meager.
		\item  $X$ is {\em meager in $R(\Delta)$} if $X \cap R(\Delta)$
		      is $\Delta$-meager.
		\item  $X$ is {\em comeager in $R(\Delta)$} if $X^c$ is meager in $R(\Delta)$.
	\end{enumerate}
\end{definition}

\noindent
The {\em resource-bounded Baire category theorem} \cite{Lutz:CMCC}
tells us that $R(\Delta)$ is not $\Delta$-meager.

\subsection{Symmetric Alternation}

Symmetric alternation is a concept in computational complexity theory introduced independently by Canetti \cite{Canetti96} and Russell and Sundaram \cite{RussellSundaram98} with different, but equivalent definitions \cite{Cai07}. We follow the definition due to Canetti.

\newcommand{\StwoTIME}{\CCfont{S}_2\CCfont{TIME}}

\begin{definition}\label{def:S2TIME_Canetti}
	Let \(T\colon\mathbb{N}\to\mathbb{N}\).
	We say that a language \(L\in \StwoTIME[T(n)]\) if there exists an \(O(T(n))\)-time verifier
	$V(x,\pi_1,\pi_2)$
	that takes
	$
		x\in\{0,1\}^n$ and $\pi_1,\pi_2\in\{0,1\}^{T(n)}$
	as input, satisfying
	\[(\exists \pi_1)(\forall \pi_2) V(x,\pi_1,\pi_2) = \chi_L(x)\]
	and  \[(\exists \pi_2)(\forall \pi_1) V(x,\pi_1,\pi_2) = \chi_L(x).\]
\end{definition}

\begin{definition}
	We define the symmetric alternation complexity classes
	\[\SPtwo = \bigcup_{c = 1}^\infty\StwoTIME[n^c].
	\]
	and \[
		\SEtwo = \bigcup_{c = 1}^\infty\StwoTIME[2^{cn}].
	\]
\end{definition}

Symmetric alternation is used to compute single-valued functions as follows \cite{Li24}.

\begin{definition}\label{def:single-valued-FS2P}
	A \emph{single‐valued} $\FSPtwo$ algorithm $A$ is specified by a polynomial $\ell(\cdot)$ together with a polynomial‐time verifier
	$V_A(x,\pi_1,\pi_2).
	$ On input $x\in\{0,1\}^*$, we say that $A$ \emph{outputs} a string $y_x\in\{0,1\}^*$ if the following two conditions hold:
	\begin{enumerate}
		\item There exists $\pi_1\in\{0,1\}^{\ell(|x|)}$ such that for every $\pi_2\in\{0,1\}^{\ell(|x|)}$, \allowbreak
		      $$V_A(x,\pi_1,\pi_2)\;\allowbreak=\;y_x.
		      $$\item There exists $\pi_2\in\{0,1\}^{\ell(|x|)}$ such that for every $\pi_1\in\{0,1\}^{\ell(|x|)}$, \allowbreak
		      $$V_A(x,\pi_1,\pi_2)\;=\;y_x.
		      $$\end{enumerate}
\end{definition}

Note that \Cref{def:S2TIME_Canetti} prescribes a single-valued $\FSPtwo$ algorithm for $\chi_L$.

Li \cite{Li24} used a single-valued $\FSPtwo$ algorithm for the range avoidance problem \cite{Korten22,ChenHiraharaRen24} to prove \Cref{th:Li_intro} that $\SEtwo$ requires exponential-size circuits.
The following is a corollary to Li's proof.

\begin{theorem_cite}{Li \cite{Li24}}\label{th:simple}
	There is an $\FSPtwo$ algorithm $\calA$ that on input $^{2^n}$ outputs a truth-table $f \in \{0,1\}^{2^n}$ such that $f$ has circuit complexity at least $\frac{2^n}{n}$.
\end{theorem_cite}

\begin{remark}
	The $\twonn$ bound in \Cref{th:Li_intro}, \Cref{th:simple}, and our results in Section \ref{sec:category_in_S2E} may be replaced by $\lutzbound$ for any fixed $\alpha\in(0,1)$
	\cite{Lutz:AEHNC,FraMil05,Hitchcock:CMMDCC}.  We use $\twonn$ to simplify notation.
\end{remark}

\section{Category in Symmetric Exponential Time}\label{sec:category_in_S2E}

Recall that $R(\p)=\E$ and $R(\pspace)=\ESPACE$. We now show that the result class of $\FSPtwo$ is $\SEtwo$.

\begin{lemma}\label{le:RFSPtwo_equals_SEtwo}
	$R(\FSPtwo) = \SEtwo$.
\end{lemma}
\begin{proof}
	Let $L = R(A)$ where $A$ is a single-valued $\FSPtwo$ algorithm. Let $V_A$ be as in \Cref{def:single-valued-FS2P}.
	On input $x$, let $n$ be the index such that $s_n = x$. The idea is to  start with $z_0 = \lambda$ and compute $z_{i+1} = A(z_i)$ until $|z_i| > n$. Then the answer for $x$ is $z_i[n]$. We make at most $n \leq 2^{|x|}$ calls to $A$.
	Formally, define a verifier $V$ that on input $s_n = x$ also takes two tuples of proofs $t_1=\pair{\pi_1^{(0)},\ldots,\pi_{1}^{(k-1)}}$ and
	$t_2=\pair{\pi_2^{(0)},\ldots,\pi_{2}^{(k-1)}}$. Then we compute $z_0 = \lambda$ and $z_{i+1} = V_A(z_i,\pi_1^{(i)},\pi_2^{(i)})$ for each $0 \leq i < k$.
	If $z_i \prprefix z_{i+1}$ for all $i$ and
	$|z_k| > n$, then $V(x,t_1,t_2)$ outputs $z_i[n]$. Otherwise, $V(x,t_1,t_2)$ outputs $0$.
	The total length of the proofs is polynomial in $n$, so the total proof length is $2^{O(|x|)}$. Since $V_A$ is polynomial-time, $V$ runs in exponential-time and it follows that $L \in \SEtwo$.

	Let $L \in \SEtwo$ and let $V$ be a verifier for $L$ as in \Cref{def:S2TIME_Canetti}. On input $z$, we let $n = |z|$. We use $V$ on input $s_n$ to extend $z$ by one bit. If $V(s_n,\pi_1,\pi_2) = 1$, we output $V'(z,\pi_1,\pi_2) = z1$. If $V(s_n,\pi_1,\pi_2) = 0$, we output $V'(z,\pi_1,\pi_2) = z0$. Then $V'$ specifies a single-valued $\FSPtwo$ algorithm because $|s_n|$ is logarithmic in $|z|$. We have $R(V')=L$, so $L \in R(\FSPtwo)$.
\end{proof}

We can now fully specify the symmetric alternation category notion used in this paper. We plug \Cref{le:RFSPtwo_equals_SEtwo} into \Cref{def:resource_bounded_category}:

\begin{enumerate}
	\item  $X$ is {\em $\FSPtwo$-meager} if player II has a winning
	      strategy in the game $G[X;\all,\FSPtwo]$.
	\item  $X$ is {\em $\FSPtwo$-comeager} if $X^c$ is $\FSPtwo$-meager.
	\item  $X$ is {\em meager in $\SEtwo$} if $X \cap \SEtwo$
	      is $\FSPtwo$-meager.
	\item  $X$ is {\em comeager in $\SEtwo$} if $X^c$ is meager in $\SEtwo$.
\end{enumerate}

Lutz \cite{Lutz87} proved that $\SIZE(\twonn)$ is meager in $\ESPACE$ where the strategy uses a brute force voting method inspired by Kannan \cite{Kann82}. We can replace this by Li's single-valued $\FSPtwo$ algorithm from \Cref{th:simple}.
\begin{theorem}\label{th:main_FSPtwo_meager}
	$\SIZE(\twonn)$ is $\FSPtwo$-meager.
\end{theorem}
\begin{proof}
	Let $\calA$ be the algorithm from \Cref{th:simple}.
	On input $x$, compute $n$ so that $2^{n-1} \leq |x| \leq 2^{n}-1$. Let $s(n) = \twonn$.
	Our constructor $\delta$ outputs
	$$\delta(x) = x0^{2^n-|x|-1}\calA(0^{2^n}).$$
	The $^{2^n-|x|-1}$ is to finish defining the language at length $n-1$.
	We then use $\calA$ to define the language at length $n$. Thus $|\delta(x)|=2^{n+1}-1$ and $\delta(x)$ defines a subset of $\{0,1\}^{\leq n}$.

	Let $\gamma$ be any constructor. Then for infinitely many lengths $n$, $R(\gamma,\delta) = R(\delta \circ \gamma)$ does not have an $s(n)$-size circuit at length $n$. Therefore $R(\gamma,\delta) \not\in \SIZE(s(n))$ and $\delta$ wins the Banach-Mazur game, so $\SIZE(s(n))$ is $\FSPtwo$-meager.
\end{proof}
\Cref{th:main} is a corollary of \Cref{th:main_FSPtwo_meager}, because $\FSPtwo$-meager implies meager in $\SEtwo$.
\MainTheorem*

Our title result follows immediately from \Cref{th:main}.
\begin{corollary}\label{co:title}
	The class of problems requiring exponential-size circuits is comeager in $\SEtwo$.
\end{corollary}

Another corollary is meagerness of $\SIZE(\twonn)$ in $\E^\NP$ under a derandomization hypothesis. This is because standard derandomization hypotheses \cite{KlivMe02} collapse $\FSPtwo$ to $\FP^\NP$ via derandomization of $\FZPP^\NP$.
Resource-bounded category in $\E^\NP$ is defined analogously to category in $\E$, using polynomial-time constructors that have access to an $\NP$ oracle in \Cref{def:resource_bounded_category}.

\begin{lemma}
	$R(\FP^\NP) = \E^\NP$.
\end{lemma}

\begin{enumerate}
	\item  $X$ is {\em $\FP^\NP$-meager} if player II has a winning
	      strategy in the game $G[X;\all,\FP^\NP]$.
	\item  $X$ is {\em meager in $\E^\NP$} if $X \cap \E^\NP$
	      is $\FP^\NP$-meager.
\end{enumerate}

\begin{corollary}
	If $\E^\NP$ requires $2^{\Omega(n)}$-size $\NP$-oracle circuits, then $\SIZE(\twonn)$ is meager in $\E^\NP$.
\end{corollary}
\begin{proof}
	Cai \cite{Cai07} showed that $\SPtwo \subseteq \ZPP^\NP$. It also holds that single-valued $\FSPtwo$ is contained in single-valued $\FZPP^\NP$ \cite{ChenHiraharaRen24}. Under the hypothesis, techniques of Klivans and van Melkebeek \cite{KlivMe02} derandomize $\FZPP^\NP$ to $\FP^\NP$ \cite{ChenHiraharaRen24}. Thus the constructor in the proof of \Cref{th:main_FSPtwo_meager} is in $\FP^\NP$, so $\SIZE(\twonn)$ is $\FP^\NP$-meager and is meager in $\E^\NP$.
\end{proof}

\section{Conclusion}\label{sec:conclusion}

Lutz and Mayordomo \cite{Lutz:TPRBM01} asked if $\Ppoly$ has measure 0 in a class smaller than $\DeltaEthree$, the third level of the exponential-time hierarchy. \Cref{th:main} implies $\Ppoly$ is meager in $\SEtwo$. Resource-bounded measure may also be defined in $\SEtwo$ using single-valued $\FSPtwo$ martingales. Does $\Ppoly$ have measure 0 in $\SEtwo$?\\
\ \\

\noindent{\bf Acknowledgement.} I thank an anonymous referee for helpful comments.

\bibliographystyle{plainurl}

\end{document}